\documentclass[a4paper,onecolumn, 11pt, accepted=2023-01-06]{quantumarticle}
\pdfoutput=1

\usepackage[url=false, doi=true, isbn=false, eprint=false]{biblatex}
\usepackage{amsmath,amssymb}
\usepackage{array}
\usepackage{braket}
\usepackage{tikz}
\usepackage{amsthm}
\usepackage{braket}
\usepackage{complexity}
\usepackage{color}
\usepackage{algorithm}
\usepackage{algpseudocode}
\usepackage{graphicx}
\usepackage{dcolumn}
\usepackage{bm}
\usepackage{verbatim}
\tikzset{diagram-node-light/.style={black, circle, draw=black, fill=yellow!20, minimum size=0.5cm, inner sep=0pt}}
\tikzset{diagram-node-dark/.style={black, circle, draw=black, fill=yellow!70, minimum size=0.5cm, inner sep=0pt}}
\tikzset{diagram-node-inverted/.style={white, circle, draw=black, fill=black!50, minimum size=0.5cm, inner sep=0pt}}
\newtheorem{conjecture}{Conjecture}
\newtheorem{corollary}{Corollary}
\newtheorem{assumption}{Assumption}
\newtheorem{theorem}{Theorem}
\newtheorem{definition}{Definition}
\newtheorem{lemma}{Lemma}

\usepackage[hidelinks]{hyperref}
\usepackage{multirow}

\begin{document}

\title{Amplitude Ratios and Neural Network Quantum States}
\author{Vojt\v{e}ch Havl\'{i}\v{c}ek}
\affiliation{IBM Quantum, IBM T.J. Watson Research Center}

\begin{abstract}
Neural Network Quantum States (NQS) represent quantum wavefunctions by artificial neural networks. Here we study the wavefunction access provided by NQS defined in [Science, \textbf{355}, 6325, pp. 602-606 (2017)] and relate it to results from distribution testing. This leads to improved distribution testing algorithms for such NQS. It also motivates an independent definition of a wavefunction access model: the amplitude ratio access. We compare it to sample and sample and query access models, previously considered in the study of dequantization of quantum algorithms. First, we show that the amplitude ratio access is strictly stronger than sample access. Second, we argue that the amplitude ratio access is strictly weaker than sample and query access, but also show that it retains many of its simulation capabilities. Interestingly, we only show such separation under computational assumptions. Lastly, we use the connection to distribution testing algorithms to produce an NQS with just three nodes that does not encode a valid wavefunction and cannot be sampled from.
\end{abstract}

\tableofcontents

\section{Introduction}

Machine learning algorithms have been recently used for simulating quantum many-body systems. A prominent example is a method introduced in Ref.~\cite{Carleo17}, called the \textit{Neural Network Quantum State} (NQS). NQS encodes a quantum wavefunction $\psi_\theta(v):\lbrace -1,1 \rbrace^n \rightarrow \mathbb{C}$. The encoded wavefunction is defined by complex-valued parameters $\theta = (a,b,W)$, where $a$ are visible biases of the network, $b$ are hidden biases, and $W$ are interaction weights.  The wavefunction can be accessed by sampling of the Born distribution $p_\theta(v) := | \psi_\theta(v)|^2$ and computing $\psi_{\theta}(v)$ up to an unknown and usually inaccessible normalization factor. This allows for efficient computation of ratio of the wavefunctions $\psi_\theta (v)/\psi_\theta (w)$ for two outcomes $v,w$.

\subsection{Results}

This work explores the relationship between this wavefunction access and results from conditional distribution testing. This  motivates a definition of a wavefunction access model, the amplitude ratio (AR) access. AR can be seen as a variant of the sample and query (SQ) access to wavefunction, which has seen use in analysis of quantum algorithms and dequantization (see for example \cite{Ewin18}). Here we argue that AR is a strictly weaker model than the standard SQ. Secondly, the connection to distribution testing allows us to show that there exist small NQS that cannot be sampled from, in an  analogy to results about sampling from Restricted Bolzmann Machines \cite{Long10, Martens13}. 

\paragraph{NQS and PCOND:} Given a NQS $\theta$ and two measurement outcomes $v_0, v_1 \in \lbrace -1, 1 \rbrace^n$, it is possible to efficiently compute and sample from $p_{\theta}(v |\, v = v_0 \text{ or } v = v_1)$. We observe that this instantiates the \emph{pair-cond} (PCOND) \emph{oracle} used in conditional distributon testing~\cite{Canonne12, Chakraborty13}. It is known that PCOND leads to significant improvements in \emph{query complexity}~\cite{Canonne12, Chakraborty13, Canonne19}. We show how this improvement enables computational advantages for algorithms for testing NQS Born distributions.

\paragraph{Amplitude ratio (AR) access:} We define a version of PCOND for amplitudes, the \emph{amplitude ratio \textnormal{(AR)} access}. We discuss that it can be understood as a variant  the sample and query (SQ) wavefunction access\footnote{AR can be seen as an SQ without access to normalization factors. In Ref.~\cite{Ewin18}, Tang uses SQ access to a vector without normalization but their key subroutine requires access to the normalization factor of the vector, resulting in SQ access to normalized data. See~Sec.\ref{Sec:ARvsSQ} for discussion.}  that found many uses in analysis of quantum algorithms and dequantization \cite{Ewin18}.
We compare AR to SQ with normalized queries and show that it retains many of its simulation capabilities, such as:
\begin{itemize}
    \item a query-efficient fidelity estimator.
    \item the ability to compute expectation values of sparse observables.
\end{itemize}

We give evidence that SQ is strictly stronger access model than AR and show that AR is a strictly stronger model than PCOND. 
We derive a robust version of the fidelity estimator and show how to estimate sparse observables with AR access.

\paragraph{NQS postselection gadgets:} We show how to postselect the Born distribution of an NQS by changing its network structure. We call such transformation a NQS postselection gadget and use it to give an NQS with only three nodes (and polynomially-bounded weights and biases) that does not encode a valid wavefunction. As the result implies that the NQS distribution cannot be sampled from, it can be understood as a counterpart to the best known hardness of sampling result for Restricted Boltzmann Machines, which shows that a certain distribution cannot be represented (and hence sampled from) with a polynomially-sized RBM \cite{Martens13}. We briefly discuss other possible application of the gadgets.  

\section{Background}
\subsection{Neural Network Quantum States (NQS)}
\label{Sec:NQS}
Ref.~\cite{Carleo17} proposed NQS representation of quantum wavefunctions by a hidden Markov models, largely inspired by Restricted Boltzmann Machines (RBM) \cite{Smolensky86, Hinton02}. We briefly review it here.

Let $v \in \lbrace -1, +1 \rbrace^n$ and define $\psi(v) = f_\theta(v)/Z_\theta$, where:
\begin{equation}
\begin{aligned}
    f_\theta(v) &= \sum_{h \in \lbrace -1, +1 \rbrace^m} \exp(a^\intercal v + b^\intercal h + v^\intercal W h) \\
    &= \exp(a^\intercal v) \prod_{i=1}^m 2\cosh \left( b_i + \sum_j^n v_j W_{ji}\right),
\end{aligned}
    \label{Eq:numerator}
\end{equation}
and $Z_\theta := \sqrt{\sum_v |f_\theta(v)|^2}$.
 The parameters: 
 \begin{align} \theta := (a, b, W),
 \label{Eq:params} \end{align}
 are all complex-valued; $a \in \mathbb{C}^n, b \in \mathbb{C}^m, W \in \mathbb{C}^{n \times m}$ and fully specify the model. We denote $\| \theta \|_\infty = \max \left(\|a\|_\infty, \|b\|_\infty, \|W\|_\infty \right)$ and assume that $\| \theta \|_\infty \leq \textnormal{poly}(n)$. 
  
  Note that $Z_\theta := \sqrt{\sum_v |f_\theta(v)|^2}$ sums over all configurations $v$ and that there is apriori no simple way to evaluate it. In contrast, the {numerator} $f_\theta(v)$ in Eq.~\ref{Eq:numerator} can be evaluated to machine precision\footnote{See Sec.~\ref{Sec:RelAppx}.} in polynomial time in $m$ and $n$. This follows by observing that each of the $m$ factors only depends on a $\sum_{j}^n v_i W_{ji}$ which has at most $n$ terms. See Fig.~\ref{fig:NQS}. 

\begin{figure}[h]
\centering
\begin{tikzpicture}[baseline = (current bounding box.center)]
\node[diagram-node-light] at (0, 0)   (v1) {$v_1$};
\node[diagram-node-light] at (1, 0)   (v2) {$v_2$};
\node[diagram-node-light] at (2, 0)   (v3) {$v_3$};

\node[diagram-node-light] at (3, 0)   (v4) {$v_4$};

\node[diagram-node-dark] at (0, -1)   (h1) {$h_1$};
\node[diagram-node-dark] at (1, -1)   (h2) {$h_2$};
\node[diagram-node-dark] at (2, -1)   (h3) {$h_3$};

\node[diagram-node-dark] at (3, -1)   (h4) {$h_4$};

\draw (h1) -- (v1);
\draw (h1) -- (v2);
\draw (h1) -- (v3);
\draw (h1) -- (v4);

\draw (h2) -- (v1);
\draw (h2) -- (v2);
\draw (h2) -- (v3);
\draw (h2) -- (v4);

\draw (h3) -- (v1);
\draw (h3) -- (v2);
\draw (h3) -- (v3);
\draw (h3) -- (v4);

\draw (h4) -- (v1);
\draw (h4) -- (v2);
\draw (h4) -- (v3);
\draw (h4) -- (v4);
\end{tikzpicture}

\caption{NQS is composed of hidden ($h$) and visible ($v$) nodes arranged in a bipartite graph. Every node can be in a state $+1$ or $-1$. Each edge $(i,j)$ of the graph carries a complex-valued \emph{interaction strength} $W_{ij}$ and every node carries a \emph{local field} $a_i$ (for visible) or $b_j$ (hidden). The model encodes the unnormalized wavefuntion into a ``marginal amplitude'' over hidden nodes (Eq.~\ref{Eq:numerator}). It can be sampled from using Gibbs sampling.
}
\label{fig:NQS}
\end{figure}
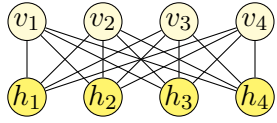
The set of parameters $\theta$ is usually found by variational optimization \cite{Carleo17}. This relies on sampling from the \emph{Born distribution} $|\psi_{\theta}(v)|^2$ and gradually updating the network parameters. Sampling from the Born distribution is usually done by ``thermalizing'' the model with a Markov Chain Monte Carlo; or \emph{Gibbs sampling} from the Born distribution. The same method is then used to sample from the Born distribution of a trained model. There is generally no guarantee that the Markov chains, both during the training and testing phase, converge to the target distribution rapidly for a given $\theta$. Remarkably, we in show in Sec.~\ref{Sec:Hardness} that there exist parameters for which no such Markov chain converges. 

\subsection{NQS and Restricted Boltzmann Machines (RBMs):}
NQS were inspired by Restricted Boltzmann Machines (RBMs)~\cite{Carleo17} and while the two models share many similarities, there are important differences. An RBM represents a \emph{distribution} $p_\theta(v)$ for $v \in \lbrace -1, +1 \rbrace^n$ by a marginal over the Gibbs-Boltzmann distribution of an Ising model on a bipartite graph. The model is, in the simplest setting, defined by a set of real-valued weights and biases (local fields). The output distribution is represented as the thermal distribution \emph{marginalized} over the hidden nodes:
\begin{align}
    p_\theta^{\textnormal{RBM}}(v) &\propto \sum_{h \in \lbrace -1, +1\rbrace^m} \exp(a^\intercal v + b^\intercal h + v^\intercal W h).
    \label{Eq:RBM_probability}
\end{align}
Even though Eq.~\ref{Eq:numerator} and Eq.~\ref{Eq:RBM_probability} look very similar, NQS is \emph{not} straightforwardly represented by a marginal distribution. The model instead implicitly defines a Born distribution:
\begin{align}
    p^{\text{NQS}}_{\theta} (v) &\propto \left| \sum_{h \in \lbrace -1, +1 \rbrace^m} \exp(a^\intercal v + b^\intercal h + v^\intercal W h) \right|^2.
    \label{Eq:Prob_NQS}
\end{align}
This allows for \emph{interference} between the summands. Because the network parameters can be complex-valued, such interference can be destructive. This means that Eq.~\ref{Eq:Prob_NQS} can evaluate to zero. In comparisson, the marginal sum for an RBM in Eq~\ref{Eq:RBM_probability} is lower-bounded by $2^m \exp(-\| a \|_\infty n)$, which is large for $m \gg n$ and small $\| a \|_\infty$. In such setting, the RBM algorithm cannot faithfully represent zeros in the output distribution. We show in Sec.~\ref{Sec:Hardness} that this implies that there exist parameters for which the NQS does not define a valid distribution.

\subsection{Approximations}
\label{Sec:RelAppx}
Here we introduce some notions of approximation that will be used throughout the work. We also state and justify the key assumption (Assumption~\ref{Ass:Key}), which ensures that the notion of amplitude ratios is well-defined.   
\paragraph{Additive approximations:}
Let $\epsilon \geq 0$. An $\epsilon$-additive approximation $\tilde{R}$ to a real number $R$ is a real number that satisfies:
\begin{align}
    |R- \tilde{R}| \leq \epsilon.
\end{align}

\paragraph{Relative approximations:} Let $\epsilon > 0$.
An $\epsilon$-relative approximation $\tilde{R}$ to a non-negative real number $R \in \mathbb{R}_0^+$ is a real number that satisfies:
  \begin{align}
      (1+\epsilon)^{-1}R \leq \tilde{R} \leq (1+\epsilon)R.
  \end{align}
  \begin{lemma}
  \label{Lemma:Ratios}
  Fix $\epsilon > 0$.
  Let $\tilde{Q} \in \mathbb{R}$ and $\tilde{R} \in \mathbb{R}$ be $\epsilon$-relative approximations to $Q,R \in \mathbb{C}$ respectively. Then $\tilde{Q}/\tilde{R}$ is a $3\epsilon$-relative approximation to $Q/R$, $\tilde{Q}\tilde{R}$ is a $3\epsilon$-relative approximation to $QR$. 
  \end{lemma}
  \begin{proof}
  First note that:
  \begin{align}
  \frac{1}{(1+\epsilon)^2}    \frac{{Q}}{{R}} \leq  \frac{\tilde{Q}}{\tilde{R}}  \leq (1+\epsilon)^2 \frac{{Q}}{{R}}.
  \end{align}
  Set $1 + \epsilon' = (1+\epsilon)^2$ as the degree of relative approximation to the ratio, from which $\epsilon' = \epsilon^2 + 2\epsilon \leq 3 \epsilon$. Given an $\epsilon$-relative approximation $\tilde{Q}$ to $Q$, $1/\tilde{Q}$ is an $\epsilon$-relative approximation to $1/Q$. This means that $\tilde{Q}\tilde{R}$ is a $3\epsilon$-relative approximation to $QR$. 
  \end{proof}

  \paragraph{Complex numbers:}
  We will require a similar notion of approximation for complex numbers. To simplify our analysis, we use the following simplification: we assume that the complex numbers are stored in polar form as $e^{i\alpha} R := (\alpha, R)$ for $\alpha \in [0,2 \pi)$ and $R \in \mathbb{R}_0^+$ and that the phase $\alpha$ is stored exactly \footnote{In practice, it will be only stored to a sufficiently high \emph{additive} precision, using up to $poly(n)$ bits. We will ignore the issues that arise by assumption of storing $\alpha$ exactly, such as when $\alpha$ is an irrational number.}. By an $\epsilon$-relative approximation to a complex number $C = e^{i\alpha}{R}$, we mean a number $\tilde{C}$ that satisfies:
  \begin{align}
      (1+\epsilon)^{-1} R \leq \tilde{R} \leq (1+\epsilon) R,
  \end{align}
  and $\alpha = \tilde{\alpha}$. Lemma~\ref{Lemma:Ratios} works for this notion of approximation if we assume that the result of adding two phases is always mapped back to $[0, 2\pi)$.

  \paragraph{Machine precision:}  Machine precision is an upper bound on the relative error $\epsilon$ due to rounding in floating point arithmetic: numbers are represented up to a finite number of bits, which leads to \emph{usually} insignificant additive errors. These errors \emph{usually} imply good relative approximation. We define this more precisely now: 
  \begin{lemma}
  \label{Lemma:AddToRel}
  Let $R \in \mathbb{R}^+$ and $\epsilon \in [0,1)$. Assuming $\epsilon \leq R$, an $\frac{\epsilon^2}{2}$-additive approximation $\tilde{R}$ to $R$ is also an $\epsilon$-relative approximation to $R$.
\end{lemma}
\begin{proof}
This follows from:
\begin{align}
   (1+\epsilon)^{-1} R &\leq  \left(1-\frac{\epsilon}{2}\right)R \leq \left(R-\frac{\epsilon^2}{2}\right) \leq \tilde{R}, \\
   \tilde{R} &\leq \left( R + \frac{\epsilon^2}{2} \right) \leq \left(1 + \frac{\epsilon}{2} \right) R \leq (1+\epsilon)R.
\end{align}
\end{proof}

  Throughout this work, we will make a simplifying assumption that by \emph{evaluation to machine precision}, we mean evaluation to $\epsilon$-relative error in $poly(\log(1/\epsilon))$ time and justify it shortly.  We will often require evaluation of some quantities to machine precision in $poly(n, \log(1/\epsilon))$ time, where $n$ is the input size to the problem. From this, we bind the error parameter $\epsilon$ to the input size and simply assume that standard functions, such as $\exp$ or $\cosh$, can be efficiently evaluated to $2^{-poly(n)}$ error for arguments with magnitude bounded by some polynomial $poly(n)$ (similarly to Eq.~\ref{Eq:numerator}). 

The following is an important consequence of Lemma~\ref{Lemma:AddToRel}:
  
  \begin{corollary}
   \label{Cor:ratio}
  If $Q, R \geq 2^{-poly(n)}$ can be efficiently evaluated to machine precision, then so can $1/Q$, $QR$ and $Q/R$.
  \end{corollary}

  \subsection{The key simplifying assumption}

  To make use of Corollary~\ref{Cor:ratio}, we will often need the following assumption regarding wavefunctions:
  \begin{assumption}
  \label{Ass:Key}
  We assume that $|\psi(v)| \geq 2^{-poly(n)}$ or $\psi(v) = 0$ for some sufficiently large polynomial $poly(n)$.
  \end{assumption}
  
  This always holds for any wavefunction represented by NQS, because for $\| \theta \|_\infty \leq poly_1(n)$, we have $f_{\theta}(v) \geq 2^{-poly_2(n)}$ for a suitable polynomial (perhaps distinct from $poly_1(n)$) and $Z \leq 2^{-n}$. The condition is not automatically guaranteed for other families of wavefunctions and this may become problematic for the access models that we study in the next section. The assumption also guarantees that a machine precision approximation to $\psi(v)$ is always representable by $poly(n)$ bits.

\section{Amplitude Ratios}
We first compare the NQS wavefunction access model to the \emph{pair-cond} (PCOND) query access to distributions from Refs.~\cite{Chakraborty13, Canonne12}. PCOND queries are strictly more powerful than sampling and we show that they can be simulated efficiently for any NQS distribution. This gives improved algorithms for distribution testing. It also leads to a modification of PCOND for quantum wavefunctions. We define this as amplitude ratio (AR) access and study it independently of NQS. We compare AR to sample and query (SQ) access used in dequantization \cite{Ewin18} and probabilistic simulation \cite{vanDenNest11}. We argue that AR is weaker access model than SQ with \emph{normalized} queries, but retains many classical simulation techniques of SQ.

\subsection{Pair-cond and NQS}
\label{Sec:PCOND}
NQS wavefunctions can be efficiently evaluated up to a normalization factor for any $v \in \lbrace -1 ,+1 \rbrace^n$ (Sec.~\ref{Sec:NQS}). This enables computation of amplitude ratios: \begin{lemma} Let $\psi(v)$ be an NQS over $\lbrace -1, +1\rbrace^n$. The ratio $\psi(i)/\psi(j)$ can be efficiently evaluated to machine precision for any $i,j\in \lbrace -1, +1 \rbrace^n$.
 \label{Lemma:PWR}
\end{lemma}

\begin{proof}
Observe that $\psi(i)/\psi(j) = f_\theta(i) / f_\theta(j)$. The claim follows from:
\begin{align}
    f_\theta(v) &= \exp(a^\intercal v) \prod_{i=1}^m 2\cosh \left( b_i + \sum_j^n v_j W_{ji}\right),
\end{align}
(see Eq.~\ref{Eq:numerator}) and noting that each of the $m$ factors only depends on $\sum_{j}^n v_j W_{ji}$, which has at most $n$ terms. Since we assumed that $\| \theta \|_\infty \leq poly(n)$ below Eq.~\ref{Eq:params}, $f_\theta(v)$ can be efficiently evaluated to machine precision (Sec.~\ref{Sec:RelAppx}).  The result follows from Corollary~\ref{Cor:ratio}.
\end{proof}

There is an analogy between amplitude comparison and \emph{pair-cond} (PCOND) oracle access used in conditional distribution testing  \cite{Canonne12}. 

\begin{definition}[PCOND \cite{Canonne12}]
Let $p$ be a probability distribution over $\Omega$. The PCOND oracle accepts an input set $S$ that is either $S = \Omega$ or $S = \lbrace i, j \rbrace$ for some $i, j \in \Omega$. It returns an element $i \in S$ with probability $p(i)/p(S)$, where $p(S) = \sum_{i \in S} p(i)$.
\label{Def:PCOND}
\end{definition}

Refs.~\cite{Canonne12, Chakraborty13} show that PCOND queries lead to significant complexity improvement for some distribution testing tasks. We observe that:

\begin{theorem}
\textnormal{PCOND} (Def.~\ref{Def:PCOND}) can be instantiated efficiently for any NQS that can be efficiently sampled from.
\label{Thm:PCOND_for_NQS}
\end{theorem}
\begin{proof}
Assume that the NQS encodes a wavefunction $\psi_\theta(v)$ over $\lbrace -1, 1\rbrace^n$. From Eq.~\ref{Eq:numerator}, we have that $\psi_\theta(v) = f_\theta(v) / Z_\theta$, where $f_\theta(v)$ can be computed efficiently to machine precision. 

To show that NQS allows PCOND queries,
fix $\Omega = \lbrace -1, 1\rbrace^n$ in the definition of PCOND.

\begin{enumerate}
    \item On input $S = \Omega$, output a sample from the Born distribution encoded by the NQS. 
    
    \item On input $S = \lbrace i, j \rbrace$ for $i,j \in \Omega$, compute: \begin{align}
\label{Eq:Flip}
r = \frac{|f_{\theta}(i)|^2}{|f_{\theta}(i)|^2 + |f_{\theta}(j)|^2} = \left(1 + \left|\frac{f_{\theta}(j)}{f_{\theta}(i)}\right|^2 \right)^{-1}.\end{align} Return $i$ or $j$ according to a $r$-biased coin flip (if both $|f_{\theta}(i)|^2$ and $|f_{\theta}(j)|^2$ are zero, return one of $i,j$ u.a.r).
\end{enumerate}
The ratio in Eq.~\ref{Eq:Flip} is not computed exactly, but we assume that the deviation from its \emph{exact} value will be only observed for exponentially many coin flips in the input size. Since we are interested in algorithms that make at most polynomialy many PCOND queries, we neglect this.
\end{proof}
Theorem~\ref{Thm:PCOND_for_NQS} enables efficient implementation of algorithms for NQS distributional testing; see Tab.~\ref{tab:Cannone_results}. 
\begin{table}[t]
\centering
    \begin{tabular}{|m{4.5cm}|m{3cm}| m{3cm}|}
    \hline 
       {\textbf{Problem}}  & {Random samples} & {PCOND queries} \\
       \hline
        Is $p_\theta(v)$ $\epsilon$-close to uniform distribution? & $\Theta(2^{n/2} / \epsilon^2)$ & $\tilde{O}(1/\epsilon^2)$ \cite{Canonne12}\\
        \hline
        Is $p_\theta(v)$ $\epsilon$-close to a known distribution $q(v)$? & $\Theta(2^{n/2} / \epsilon^2)$ & $\tilde{O}(n^4/\epsilon^4)$~\cite{Canonne12} \\
        \hline
        Are $p_{\theta_1}(v)$, $p_{\theta_2}(v)$ $\epsilon$-close? & $\Omega(2^{n/2} / \epsilon^2)$ & $\tilde{O}(n^6/\epsilon^{21})$ \cite{Canonne12} \\
        \hline
    \end{tabular}
    \caption{PCOND results of Ref.~\cite{Canonne12}. Here $\epsilon$-close means $\epsilon$-close in the total variational distance: $d(p,q)_{\text{TV}} := \frac{1}{2} \| p-q\|_1 = \frac{1}{2} \sum_x |p(v) - q(v)|$. Since Sec.~\ref{Sec:PCOND} shows that any efficiently samplable NQS has an efficient PCOND oracle, the sample complexity improvements translate to improvements in the runtime of distribution testing algorithms for NQS.}
    \label{tab:Cannone_results}
\end{table} The algorithms test the total variation distance of two NQS Born distributions and have exponential advantage in their runtime compared to sampling algorithms. 

\paragraph{Procedure compare} The key algorithmic tool used in Ref.~\cite{Canonne12} is a procedure \emph{compare} that uses PCOND queries for estimating probability ratios. {Given that the ratio $f_\theta(i)/f_\theta(j)$ is not too large or too small, the procedute compare outputs its $1/poly(n)$ additive approximation with high probability of success using $poly(n)$ many PCOND queries \cite{Canonne12}.}  For NQS, this ratio can be efficiently computed to machine precision, which offers further simplifications of the algorithms. This motivates the definition of \emph{amplitude ratio} wavefunction access model introduced in the next section.

\paragraph{PCOND and RBM:} PCOND can be also instantiated for RBMs. The PCOND algorithms of  Ref.~\cite{Canonne12} apply with little modification to RBMs and imply significant speedups for some distribution testing tasks. Results of Ref.~\cite{Canonne12} were presented as a part of a theoretical analysis in conditional property testing and do not have runtimes that would make them immediately practical. No natural and efficient instantiation of the PCOND oracle (such as the one in the context of RBMs) has been known to the authors in mid 2021~\cite{cannone_pc}. It is therefore possible that the algorithms could be optimized and perhaps used in practice.

\paragraph{Previous quantum information work on PCOND:}
PCOND access was studied from the quantum computing viewpoint by Sardharwalla, Strelchuk and Jozsa in Ref.~\cite{Sardharwalla16}, but the AR definition (introduced below) is new. The difference between their definition and AR is that their quantum extension of PCOND oracle (PQCOND) provides access to conditional probabilities associated with the underlying distribution -- their oracles are defined at a level of \emph{quantum states}, while the AR access is always defined with respect to some fixed basis/\emph{wavefunction}. The authors give a PQCOND version of the \emph{compare} procedure from Ref.~\cite{Canonne12} and show that PQCOND queries have polynomial improvements upon many of the PCOND distribution testing results. They also derive results on boolean distribution testing and quantum spectrum testing. 
There does not seem to be an obvious connection between AR access and PQCOND, but it would be interesting to understand this better.

\subsection{Amplitude Ratio (AR) Access}
Because NQS allows for computation of \emph{amplitude ratios} to high precision, they should offer stronger access to the wavefunction than PCOND. We study this as a wavefunction access model, independently of NQS. 
\begin{definition}[Exact AR]
\label{Def:AR_exact}
Let $\psi: \Omega \rightarrow \mathbb{C}$ be a wavefunction over $\Omega$. The $\textnormal{AR}$ oracle accepts as an input either an ordered pair $S = ( i, j )$ or $S = \Omega$. If $S = \Omega$, it returns a random sample from the Born distribution $|\psi(v)|^2$ over $v \in \lbrace -1, 1 \rbrace^n$. 
If $S = (i, j)$, it returns $\psi(i)/\psi(j)$. If the ratio diverges, it returns a special symbol `DIV' and it conventionally returns $1$ for $\psi(i) = \psi(j) = 0$.
\end{definition}

Definition~\ref{Def:AR_exact} assumes that the queries are returned exactly. This becomes problematic if $\psi(i)/\psi(j)$ becomes irrational number as the query result is an infinitelly-long output. This issue can be dealt with as follows:

\begin{definition}[AR]
Let $\psi: \Omega \rightarrow \mathbb{C}$ be a wavefunction over $\Omega$ subject to Assumption~\ref{Ass:Key}. For $\epsilon \in [0, 1)$, the $\textnormal{AR}(\epsilon)$ oracle accepts as an input either an ordered pair $S = (i, j )$ or $S = \Omega$. If $S = \Omega$, it returns a random sample from the Born distribution $|\psi(v)|^2$ over $v \in \lbrace -1, 1 \rbrace^n$. 
If $S = ( i, j )$, it returns an $\epsilon$-relative approximation to $\psi(i)/\psi(j)$.\footnote{To make the queries deterministic, one can force this to be the approximation with the shortest binary expansion.} If the ratio diverges, it returns `DIV' and it conventionally returns $1$ for $\psi(i) = \psi(j) = 0$.
We write AR := AR($\epsilon$) if $\epsilon$ scales as $2^{-poly(n)}$ for some polynomial $poly(n)$.
\footnote{The AR oracle can be defined rigorously in the oracle Turing machine model. Let the author know if you ever need this.}
\label{Def:AR_REL}
\end{definition}

\subsection{SQ vs AR}
\label{Sec:Comparing}
We now compare AR with another type of access to quantum wavefunctions, the \emph{sample and query} (SQ) access. SQ was defined in Ref.~\cite{Ewin18}, but was previously used as \emph{computational tractability} with additional computational requirements in Ref.~\cite{vanDenNest11}. There is a difference between the two definitions: Ref.~\cite{Ewin18} defines SQ access \emph{without} normalizing the queries, but subsequenly assumes knowledge of the normalization factor (see for example Prop 4.2.~\cite{Ewin18}), while Ref.~\cite{vanDenNest11} assumes normalization (as well as efficiency) in the definition of computationally tractable states (Ref.~\cite{vanDenNest11}, Def.~1). Tang's definition of SQ also does not treat the underlying object as a wavefunction, but more generally as a real valued vector with $\ell_2$-norm sampling access. This presentation puts less emphasis on the need for the normalization factor. Here we use the SQ access model assuming that \emph{the queries yield normalized amplitudes}, but impose no efficiency constraints.  

\begin{definition}[Exact SQ]
Let $\psi: \Omega \rightarrow \mathbb{C}$ be a wavefunction over $\Omega$. The wavefunction has $\textnormal{SQ}$ access if $\psi(i)$ can be computed for any $i \in \lbrace -1, +1 \rbrace^n$ and its Born distribution $|\psi(i)|^2$ can be sampled from.
\end{definition}

The above definition has the same problem as \emph{exact} AR (Def.~\ref{Def:AR_exact}): if the result of the query is an irrational number, the result of the query will not have bounded size. 
We update it as follows:
\begin{definition}[SQ]
\label{Def:Relative_SQ}
Let $\psi: \Omega \rightarrow \mathbb{C}$ be a wavefunction over $\Omega$ subject to Assumption~\ref{Ass:Key}. For $\epsilon \in [0, 1)$, the $\textnormal{SQ}(\epsilon)$ oracle accepts as an input either a singleton set $S = \lbrace i \rbrace$ or $S = \Omega$. If $S = \Omega$, it returns a random sample from the Born distribution $|\psi(v)|^2$ over $v \in \lbrace -1, 1 \rbrace^n$. 
If $S = \lbrace i \rbrace$, it returns an $\epsilon$-relative approximation to $\psi(i)$ with the shortest binary expansion.
We write SQ := SQ($\epsilon$) if $\epsilon$ scales as $2^{-poly(n)}$ for some polynomial $poly(n)$.
\end{definition}

\paragraph{AR and unnormalized SQ}
\label{Sec:ARvsSQ}
Tang used SQ with unnormalized queries in \cite{Ewin18} but subsequently assumed knowledge of the normalizing factor to the wavefunction throughout the work. Without knowledge of the normalization factor, the only information about the vector with such access comes from amplitude ratios and sampling. This way, the definition is essentially the same as AR -- so why bother with AR? 

The key reason for defining AR is to emphasize that the normalization factor of the wavefunction is simply unavailable, aside from its empirical estimate by sampling. It seems problematic to guarantee this (at least somewhat) rigorously: does evaluation of the state up to an ``arbitrary'' normalization factor always force you to not use it? \emph{That subtlety aside, one can also see AR as unnormalized SQ access and the rest of this section as comparison between normalized and unnormalized variants of SQ.}

\paragraph{AR is not stronger than SQ} AR can be simulated by SQ:
\begin{lemma} Under Assumption~\ref{Ass:Key}, an AR($\epsilon$) query can be simulated by two SQ(${\epsilon}/3$) queries.
\label{Lemma:AR_equals_SQ}
\end{lemma}
\begin{proof}
Let $\tilde{\psi}(i)$ and $\tilde{\psi}(j)$ be results of querying SQ($\epsilon/3$) on $i$ and $j$ respectivelly. To simulate AR($\epsilon$), output $\tilde{\psi}(i)/\tilde{\psi}(j)$. This is an $\epsilon$-approximation to ${\psi}(i)/{\psi}(j)$  by Lemma~\ref{Lemma:Ratios}. 
\end{proof}

\paragraph{Evidence that AR is weaker than SQ:}
To show that AR access is \emph{in some sense} weaker than SQ, we show conditional separation between variants of the two models under \emph{efficiency} constraints. SQ is related to the concept of computationally tractable (CT) states~\cite{vanDenNest11}:
\begin{definition}[Computationally Tractable (CT) states]
An $n$-qubit wavefunction $\psi(v)$ over $v\in \lbrace-1,+1\rbrace^n$, subject to Assumption~\ref{Ass:Key}, is computationally tractable (CT) if both queries in Def.~\ref{Def:Relative_SQ} can be implemented with a $poly(n)$-time randomized algorithm.\footnote{Van den Nest uses exact queries, but also writes that his definition generalizes straightforwardly to computation with exponential additive precision. This definition is a modification of Def.~1, Ref.~\cite{vanDenNest11}, which also does not need the Assumption~\ref{Ass:Key}}

\end{definition}
This can be seen as SQ with efficient classical sampler and efficient classical algorithm for the amplitude queries. The key capability of CT states is an efficient algorithm that, given two CT wavefunctions $\psi(v)$ and $\phi(v)$, approximates $\braket{\psi |\phi}$ in polynomial time to inverse-polynomial precision (see Theorem~3. and Lemma~3. in Ref.~\cite{vanDenNest11} or equivalently Prop~4.8. in Ref.~\cite{Ewin18}). This enables estimation of constant-local bounded observables on CT states or simulation of sparse quantum circuits. Techniques related to CT framework were used in dequantization algorithms in Ref.~\cite{Ewin18} or quantum algorithm analysis in Ref.~\cite{Havlicek18}. 

Analogously to CT states, we define  \emph{amplitude ratio \textnormal{(AR)} states} and show that their \emph{fidelities} and expectation values on constant-local observables can be also efficiently computed. They subsume CT states, which suggests that the CT requirement can be relaxed to AR in many applications.
\begin{definition}[Amplitude Ratio (AR) states]
An $n$-qubit wavefunction $\psi(v)$ over $v\in \lbrace-1,+1\rbrace^n$, subject to Assumption~\ref{Ass:Key}, is an amplitude ratio (AR) state if queries to AR in Def.~\ref{Def:AR_REL} can be implemented with a $poly(n)$-time randomized algorithm.
\end{definition}

We show that, subject to Assumption~\ref{Ass:Key}, all CT states are AR states. Our \emph{evidence} that SQ access is somewhat stronger than the AR access follows from that for \emph{machine precision}, not all AR states are CT, unless $\#\P = \textsf{FBPP}$.\footnote{$\#\P$ is the class of function problems that enumerate accepting paths to $\NP$ problems \cite{Valiant79}. $\textsf{FBPP}$ is the class of functions computable in polynomial time with a probabilistic Turing machine.}  This is shown using the separation between (exact) counting and uniform sampling by Jerrum, Valiant and Vazirani (Sec.~4 \cite{Jerrum86}).

\begin{theorem}
\label{Thm:CT_AR}
All CT states are AR. There exist AR states that are not CT unless $\# \P = \textnormal{\textsf{FBPP}}$.
\end{theorem}
\begin{proof}
(All CT states are AR): Both models perform AR and SQ queries in $poly(n)$ time. The rest follows from Lemma~\ref{Lemma:AR_equals_SQ}.

\medskip

(There is an exact AR state that is not CT):
The proof uses the observation that uniform sampling of solutions to a given boolean formula over $n$ variables in disjunctive normal form (DNF) is easy, while their exact enumeration is $\# \P$-complete (Sec.~4 of Ref.~\cite{Jerrum85}).  
We consider a quantum state that is a uniform superposition over satisfying assignments to the boolean formula and show that its Born distribution can be easily sampled from. The normalization factor of such state counts the number of satisfying solutions to the formula, which is $\# \P$-complete to compute exactly. We show that a good relative approximation to the amplitude determines this quantity.

Given a boolean formula in DNF over $n$ variables, interpret $v \in \lbrace -1, +1\rbrace^n$ as an assignment to its variables: if $v_i = +1$, then the $i$-th variable is true and if $v_i = -1$, the $i$-th variable is false. Define the DNF formula-``state'' as follows:
\begin{align}
    \psi_{\textnormal{DNF}}(v) &= \begin{cases} \frac{1}{\sqrt{Z}} \textnormal{ if } v \text{ satisfies the formula}. \\
    0 \textnormal{ otherwise.}
    \end{cases}
\end{align}
Notice that $Z$ is the number of variable assignments that satisfy the formula. 

For any input $v$, the predicate $|\psi_{\textnormal{DNF}}(v)| > 0$ can be tested by plugging in the variable assignments into the formula. Any non-zero amplitude evaluates to $1/\sqrt{Z}$, from which it is possible to compute the amplitude ratio as required by exact AR. 

The Born distribution $|\psi_{\textnormal{DNF}}(v)|^2$ is a uniform distribution over the satisfiying assignments to the boolean formula. This distribution can be sampled from exactly by a polynomial-time randomized algorithm described in Sec.~4 of Ref.~\cite{Jerrum86}. It works as follows: For a DNF boolean formula $F = F_1 \vee F_2 \ldots \vee F_m$ in $n$ variables, where $F_i$ is a conjunction of literals for all $i \in [1, m]$, let $S_j \subseteq \lbrace -1, +1 \rbrace^n$ be the set of satisfying assignments to $F_j$. Note that $|S_j|$ is easy to compute, because $F_j$ is only satisfied if it fixes a subset of the variables involved in it and the remaining variables can take arbitrary values. Let $S = \bigcup_j S_j$ the set of all satisfying assignments to $F$. The aim is to sample uniformly over $S$, which is achieved by Algorithm~\ref{alg:DNF}.

\begin{algorithm}
\caption{Uniform sampling of DNF assignments \cite{Jerrum86}}\label{alg:DNF}
\begin{algorithmic}
\For{$i \in 1 \ldots m$}: 
   \State $j \gets$ random integer from $[1,m]$ chosen proportionally to $|S_j|/\sum_k |S_k|$. 
   \State $a \gets$ random element of $S_j$.
   \State $N \gets |\lbrace k \in [1,m] : a \in S_k \rbrace|$.
   \State With $1/N$ probability, output $a$ and halt.
\EndFor
\end{algorithmic}
\end{algorithm}

The algorithm does not halt with small probability. If it does not halt, rerun the for-loop. The output of the algorithm is a uniformly random satisfying variable assignment. This implies that $\psi_{\textnormal{DNF}}(v)$ is an exact-AR state.\footnote{As discussed in Ref.~\cite{Jerrum86} Sec.~4, Algorithm 1 cannot be implemented on a probabilistic Turing Machine. A modification of it, also in \cite{Jerrum86} Sec.~4, deals with this subtlety and leads to an $\textsf{FBPP}$ algorithm.}

The family of states $\psi_\text{DNF}$ is not exact-CT unless $\#\P = \textsf{FBPP}$, because $\psi_\text{DNF}(v) = 1/\sqrt{Z}$ for any satisfying variable assignment $v$. Evaluating this exactly is a $\#\P$-complete problem, because $Z$ is the number of satisfying solutions to the boolean formula. So unless $\textsf{FBPP} = \#\P$, $\psi_\text{DNF}(v)$ is an exact-AR state that is not exact-CT.

It remains to show that $\psi_\text{DNF}(v)$ is not CT. To do this, we show that CT can compute  $Z$ to sufficient accuracy to determine it exactly. First note that $Z \leq 2^n$. We want to choose $\epsilon$ so that the $\epsilon$-relative approximations to $Z$ and $Z+1$ can be distinguished for any $0\leq Z \leq 2^n$. This happens if $(1+\epsilon)Z < (1+\epsilon)^{-1}(Z+1)$. Since $\epsilon < 1$, we have that $(1+\epsilon)^2 < 1+3\epsilon$ and $(1+3\epsilon)Z < Z + 1 \implies (1+\epsilon)^2Z < Z+1$. We can therefore choose $\epsilon \leq 2^{n-2}$ to recover $Z$ exactly. It follows that if $\psi_{\textnormal{DNF}}(v)$ was CT, $Z$ could be computed exactly by a polynomial time randomized algorithm ($\textsf{FBPP}$). This problem is however $\#\P$-complete, so this gives a contradiction unless $\#\P$ = $\textsf{FBPP}$.
\end{proof}

We remark that the above argument won't work for $\epsilon = 1/poly(n)$.  The reason is that there is a polynomial time randomized algorithm for approximating the number of DNF satisfying clauses to relative error $\epsilon$ that outputs the solution with high probability in $poly(n, 1/\epsilon)$ time \cite{Karp89}.\footnote{Such algorithm is called \emph{fully-polynomial randomized approximation scheme} (FPRAS)} This is worse approximation than what is assumed in the definition of AR states, which shows that one has to use different argument to separate CT($\epsilon$) from AR($\epsilon$) for $\epsilon$ scaling as $1/poly(n)$. Similar argument to the above leads to an argument that \emph{could} separate SQ and AR in terms of their query complexity:

\begin{lemma}
Given SQ access to $\psi_{\textsf{DNF}}$ from the previous theorem, a single query suffices to approximate $Z$ to machine precision. 
\end{lemma}
\begin{proof} (sketch)
Given $\psi_{\textnormal{DNF}}$ and an SQ access to it, the SQ algorithm draws a sample $v$ from $|\psi_{\textnormal{DNF}}(v)|^2$ and evaluates $\psi_\textnormal{DNF}(v) = 1/\sqrt{Z}$, which gives a machine precision approximation to $Z$. 
\end{proof}

\begin{conjecture}
$Z$ for $\psi_{\textnormal{DNF}}$ can be query-efficiently estimated by AR to no better than $1/poly(n)$ relative error.
\label{Conj:Z}
\end{conjecture}
The reasoning behind the conjecture is the following: 
Given AR access to $\psi_{\textnormal{DNF}}$, any amplitude ratio query on a pair of nonzero amplitudes evaluates to $1$. Any amplitude ratio on pairs of zero and nonzero amplitudes gives either $0$ or `DIV'. By finding a $v$ with zero amplitude and a single $w$ with a nonzero amplitude, AR can detect if $\psi_\textnormal{DNF}(z) > 0$ for any $z\in \lbrace -1, +1 \rbrace^n$. Possibly the best way to approximate $Z$ with this access \emph{seems to be} variants of importance/nested sampling (see for example \cite{Karp89, Stefankovic06, Dyer91, Huber12}), which at best lead to $poly(n)$-sample algorithm that estimates $Z$ to $1/poly(n)$ relative error.\footnote{FPRAS runs in $poly(n, 1/\epsilon)$, but we assume $\epsilon = 1/poly(n)$ binding. An AR access algorithm that achieves this precision for this problem is a variant of the Karp-Luby-Madras Markov Chain \cite{Karp89}.} It would be extremely surprising if these methods were not asymptotically optimal. I don't have a proof though.

\begin{theorem}
Assuming Conjecture~\ref{Conj:Z}, there is a task that requires just one SQ query, but at least $poly(n)$-many AR queries.
\label{Thm:AR_vs_SQ}
\end{theorem}

\paragraph{AR is stronger than sample access}
This is shown by noting that AR access implies PCOND access and then referencing the  known results from conditional distribution testing that separate PCOND and sample access.

\begin{lemma}
AR simulates \textnormal{PCOND}.
\end{lemma}
\begin{proof}
See discussion above Eq.~\ref{Eq:Flip}.
\end{proof}
 AR access is stronger than PCOND access and PCOND is strictly stronger than sample access (Tab~\ref{tab:Cannone_results}). It is worth noting that the algorithms of \cite{cannone_pc} work even if the AR ratios can be approximated only to $1/poly(n)$ additive error. It may be therefore interesting to study weaker variants of AR. 
 
 \paragraph{AR is stronger than PCOND}
 
 \begin{lemma}
 AR is stronger than \textnormal{PCOND}.
 \label{Lemma:AR_vs_PCOND}
 \end{lemma}
 \begin{proof} (sketch)
 AR can decide if $|\psi(i)/\psi(j)|^2 \leq 2^{-n}$ in a single query, but the same task requires exponentially many PCOND queries. This is because the PCOND model can estimate the ratio only by sampling, which is limited by the usual concentration bounds. These imply a lower bound of exponentially many queries.
  \end{proof}
  
 One can object that the above comparison is rather unfair, because AR computes the amplitude ratio to a high degree of precision with a single query, while PCOND can only do so query-efficiently to $1/poly(n)$ additive error. We strengthen the above lemma to separation from a version of high-precision PCOND that can query for ratios of Born distribution as follows:
 \begin{lemma}
    AR is stronger than a version of \textnormal{PCOND} that can query ratios of Born distributions to machine precision.
    \label{Lemma:PCOND_vs_AR}
 \end{lemma}
 \begin{proof}
 First, this version of PCOND can be simulated by AR by squaring all amplitude ratio queries. 
 Let $\psi = \frac{1}{\sqrt{2}}\left( \begin{array}{c} 1 \\ -1 \end{array} \right)$ and $\phi = \frac{1}{\sqrt{2}}\left( \begin{array}{c} 1 \\ 1 \end{array} \right)$ and note that: 
 \begin{align}
     |\braket{\psi|\phi}|^2 &= \sum_{y,x} \psi^*(x)\psi(y) \phi^*(y) \phi(x) = \sum_{x,y} |\psi(x)|^2 |\phi(y)|^2 \frac{\psi(y)\phi(x)}{\psi(x)\phi(y)} \\
     &= \frac{1}{4} \left[\frac{\psi(0)}{\phi(0)}\frac{\phi(0)}{\psi(0)} + \frac{\psi(1)}{\phi(1)}\frac{\phi(0)}{\psi(0)}+ \frac{\psi(0)\phi(1)}{\phi(0)\psi(1)}+\frac{\psi(1)\phi(1)}{\phi(1)\psi(1)} \right] = 0,
 \end{align}
 which can be computed with $2$ AR queries. The states have the same, uniform, Born distributions, so all PCOND ratio queries evaluate to $1$. The PCOND therefore cannot evaluate the overlap of $\psi$ and $\phi$. The PCOND access model therefore cannot evaluate the overlap of $\psi$ and $\phi$.
 \end{proof}

 \paragraph{Summary}
Results of Sec.~\ref{Sec:ARvsSQ} can be summarized as :
\begin{align}
    \textnormal{SQ} \overset{*}{\supset} \textnormal{AR} \supset \textnormal{PCOND} \supset \textnormal{SAMP},
\end{align}
since PCOND access is separated from sampling access to a Born distribution of a wavefunction by Tab~\ref{tab:Cannone_results}, AR is separated from PCOND access by Lemma~\ref{Lemma:PCOND_vs_AR}, and SQ is \emph{almost surely} (that's why the star) separated from AR by Theorems~\ref{Thm:CT_AR} and~\ref{Thm:AR_vs_SQ}. By separation, we mean that there exists at least one problem that can be query efficiently solved by one of the classes, but not by the other.

\subsection{AR and probabilistic simulation}
We now show that AR states retain many simulation capabilities of CT states. Most of the results follow from standard algorithms used with NQS that implicitly used the AR access. We improve some of them, for example by robustification of the AR fidelity estimator, to give the closest possible analogues to the previous results for CT states \cite{vanDenNest11}. 

\label{Sec:AR_fidelity}
\paragraph{AR fidelity estimator:}
Given two amplitude ratio (AR) states $\psi$ and $\phi$, there is a randomized algorithm that approximates their fidelity $|\braket{\psi | \phi}|^2$ in polynomial time to inverse-polynomial precision. Medvidovic and Carleo use the following estimator for NQS in Ref.~\cite{Medvidovic21}: 
\begin{align}
    |\braket{\phi|\psi}|^2 &= \sum_{x,y} \phi^*(x) \psi(x)\psi^*(y)\phi(y)  \\
    &= 
    \frac{Z_\phi}{Z_\psi} \frac{Z_\psi}{Z_\phi}\sum_x |\phi(x)|^2 \frac{f_\psi(x)}{f_\phi(x)} \sum_y |\psi(y)|^2 \frac{f_\phi(y)}{f_\psi(y)} \\
    &= \sum_x |\phi(x)|^2 \frac{f_\psi(x)}{f_\phi(x)} \sum_y |\psi(y)|^2 \frac{f_\phi(y)}{f_\psi(y)}.
\end{align}
Every term in the summation uses two AR ratio queries, which can be seen from:
\begin{align}
     | \braket{\phi|\psi}|^2 &= \sum_{x,y} |\phi(x)|^2 |\psi(y)|^2 \left[\frac{f_\psi(x)}{f_\psi(y)}\right] \left[ \frac{f_\phi(y)}{f_\phi(x)}\right].
\end{align} This can be operationally understood as follows: sample $(x,y) \sim |\phi(x)|^2 |\psi(y)|^2$ (which is a valid product distribution) and compute the product of AR ratios on $\psi$ for $(x,y)$ and $\phi$ for $(y,x)$. 

We now show that the estimator has finite variance and give a robust version of it with fast concentration around the mean. Set:
\begin{align}
\label{Eq:overlap_rv}
    G(x,y) = \left[\frac{\psi(x)\phi(y)}{\psi(y) \phi(x)}\right] = \left[\frac{f_\psi(x)f_\phi(y)}{f_\psi(y)f_\phi(x)}\right] ,
\end{align}
and notice that:
\begin{align}
    \mathbb{E}[|G|^2] &=\sum_{x,y} |\phi(x)\psi(y)|^2 |G(x,y)|^2 = 1.
\end{align}
From this, the variance becomes:
\begin{equation}
\begin{aligned}
    \sigma^2[G] &= \mathbb{E}[|G|^2] - \mathbb{E}[|G|]^2 \leq 1.
\end{aligned}
\end{equation}

While the main utility of $G(x,y)$ is a simplification of the variance analysis, it could also allow for additional cancellation between $f_\psi(x)$ and $f_\phi(y)$ that could not be exploited in the product of means estimator of Ref.~\cite{Medvidovic21}. 

\paragraph{Robust AR Fidelity Estimator:}Despite the fact that the estimator has a finite variance, 
the random variable $G(x,y)$ is unbounded because it contains a ratio of wavefunctions evaluated at two distinct points and may explode if either $f_\psi(y)$ or $f_\phi(x)$ are close to zero. If we are unlucky enough to obtain such outlier in our empirical estimation, it can significantly skew the statistics. It is therefore desirable to use an estimator that is less affected by outliers -- such estimators are often called \emph{robust}. We show how to make the above estimator robust using the median of means amplification. 
\begin{theorem}[Median of means estimator]
\label{Thm:MoM}
Let $k, \ell$ be two integers. Define the empirical mean $\bar{G}$ as: $\bar{G} = \sum_i G_i / k$. Compute $\ell$ such empirical means and use the median as the estimator. Then:
{\small \begin{align}
    \textnormal{Pr} \left( \left| \textnormal{med}[\lbrace\bar{G}\rbrace_{i=1}^\ell] - \mathbb{E}[G] \right| > \epsilon  \right) \leq \exp \left( -2\ell \left[ \frac{1}{2} - \frac{\sigma^2}{k\epsilon^2} \right]^2 \right) 
    \label{Eq:Mom}
\end{align}}

\end{theorem}
\begin{proof}
We have $\sigma^2(\bar{G}) = \sigma^2[{G}]/k$. By Chebyshev inequality:
\begin{align}
    p_* =: \text{Pr} \left( \left| \bar{G} - \mathbb{E} [G] \right| \geq \epsilon \right) \leq \frac{\sigma^2}{k \epsilon^2} := p.
    \label{Eq:Chebyshev}
\end{align}
The median of means condition in Eq.~\ref{Eq:Mom} is violated if the majority of empirical means violate the Chebyshev condition in Eq.~\ref{Eq:Chebyshev}. The probability of this happening is at most:
\begin{align}
     \textnormal{Pr} \left( \left| \textnormal{med}[\lbrace\bar{G}\rbrace_{i=1}^\ell] - \mathbb{E}[G] \right| > \epsilon  \right) &= \sum_{k = \lceil \frac{\ell}{2} \rceil}^{\ell} {\ell \choose k} p_*^k (1-p_*)^{\ell - k} \\ &\leq \sum_{k=\lceil \frac{\ell}{2} \rceil }^\ell {\ell \choose k} p^k (1-p)^{\ell - k},
\end{align}
where the inequality follows by monotonicity.
We can bound this by: 
\begin{align}
    \sum_{k=\lceil \frac{\ell}{2} \rceil }^\ell {\ell \choose k} p^k (1-p)^{\ell - k} &\leq \exp\left(-2\ell \left(p-\frac{1}{2} \right)^2\right),
\end{align}
where we used a tail-bound on the binomial distribution.  
\end{proof}

This technique is commonly used in computer science (see Eg.~\cite{Jerrum86} or a recent review~\cite{Lerasle19}) and was previously used with SQ access/CT states by Tang \cite{Ewin18}, where it was presented as a standard technique. An alternative way to make the SQ estimators robust was used by Van den Nest in Ref.~\cite{vanDenNest11}, but it does not work for AR. A corollary of Theorem~\ref{Thm:MoM} is a polynomial-time robust estimator of the overlap:
\begin{corollary}[Robust AR fidelity estimator]
\label{Cor:robust_fidelity}
 There is an algorithm that estimates $|\braket{\psi | \phi}|^2$ to $\epsilon$-additive error with probability $1-e^{-n}$ using $64n/\epsilon^2$ AR queries.
\end{corollary}
\begin{proof}
Use the median of means estimator in Theorem~\ref{Thm:MoM}. One evaluation of the random variable $G(x,y)$ in Eq.~\ref{Eq:overlap_rv} costs two AR queries. Use the median of $\ell$ empirical means of $G(x,y)$ (each over $k$ evaluations of the random variable) as the estimator of the overlap. Setting $k = 4/\epsilon^2$ and $\ell = 8n$ in Theorem~\ref{Thm:MoM} gives: 
\begin{align}
    \exp\left(-16n\left( \frac{1}{2} - \frac{\sigma^2}{4}\right)^2 \right) < \exp(-n),
\end{align}
which follows from: 
\begin{align}
    \left(\frac{1}{2} - \frac{\sigma^2}{4}\right)^2 \in \left[\frac{1}{16}, \frac{1}{4}\right].
\end{align}
The overall number of AR queries that achieves this is at most $2\ell k = 64n/\epsilon^2$.
\end{proof}

We briefly compare this estimator to the CT estimator of Ref.~\cite{vanDenNest11}. The algorithms achieve the same goal and their asymptotic query complexity in $\epsilon$ and $n$ are the same. The above fidelity estimator however does not require computation of the amplitude, but only amplitude ratios, which is a computationally easier problem as argued in Theorem~\ref{Thm:CT_AR}.

\paragraph{Estimating Sparse Observables:}
\label{Sec:Estimating_local_observables}
Given AR access, there is an algorithm for estimating expectation value of a (hermitian) observable $O$ expanded in the same basis as the wavefunction. This algorithm is well-known in computational physics as local observable estimation:
\begin{align}
    \braket{\psi|O|\psi} = \sum_k \sum_j \psi(j) \psi^*(j) O_{jk} \frac{\psi(k)}{\psi(j)} = \sum_j |\psi(j)|^2 \left[\sum_k O_{jk} \frac{f_\psi(k)}{f_\psi(j)} \right].
\end{align}
This can be interpreted as sampling from the Born distribution $|\psi(j)|^2$ and querying the $f_\psi(k)/f_\psi(j)$ AR ratios for every $k$ that appears in:
\begin{align}X := \sum_k O_{jk} \left[{f_\psi(k)}/f_\psi(j)\right].\end{align} This means that if $O$ has at most $poly(n)$ non-zero entries in each row, we can compute this estimator in polynomial time. As previously, we bound the variance of this estimator. We have that:
\begin{equation}
    \begin{aligned}
       \sigma^2[X] = \mathbb{E}[|X|^2] - \mathbb{E}[|X|]^2 \leq \mathbb{E}[|X|^2],
    \end{aligned}
\end{equation}
and 
\begin{align}
    \mathbb{E}[|X|^2] &= \sum_j |\psi(j)|^2 \sum_{k,\ell} O_{jk} O^*_{j \ell} \frac{\psi(k)}{\psi(j)}\frac{\psi^*(\ell)}{\psi^*(j)} \\
    &\leq \sum_{j,k,\ell} O_{jk} O^*_{j \ell} \psi(k) \psi^*(\ell) \\
    &= \sum_{k,\ell} {O^2}_{\ell k} \psi^*(\ell) \psi(k) = \braket{\psi|O^2|\psi}.
\end{align}

\footnote{
The cancellation of $|\psi(j)|^2$ is problematic for all values for which $\psi(j) =0$, because the random variable becomes unbounded on the values outside of the support of $|\psi(j)|^2$. The expectation value then depends on the values that the observable $O_{jk}$ takes on samples outside of the support of $|\psi(j)|^2$, which is undesirable. Let $\Omega$ be the domain of $|\psi(j)|^2$ and $\Sigma := \lbrace j \in \Omega : |\psi(j)|^2 > 0 \rbrace $ be its support. To alleviate the above issue, it is more natural to define: \begin{align}
    \mathbb{E}[|X|^2]' &:= \sum_{j \in \Sigma} |\psi(j)|^2 \left| \sum_{k \in \Omega}  \frac{\psi_k}{\psi(j)} O_{jk}\right|^2,
\end{align}
such that the cancellation of $|\psi(j)|^2$ is well-defined. We can write this as:
\begin{align}
       \mathbb{E}[|X|^2]' &= \sum_{j \in \Sigma} \bigg( \underbrace{ \sum_{k \in \Omega} \psi_k O_{jk} }_{\phi_j}\bigg)\bigg( \underbrace{\sum_{\ell \in \Omega} \psi_\ell O_{j\ell}}_{\phi_j^*} \bigg)^*  = \sum_{j \in \Sigma} |\phi_j|^2 \\ &\leq \sum_{j \in \Omega} |\phi_j|^2  = \braket{\psi|O^2|\psi}.
\end{align}
Irrespectivelly of the definition of $\mathbb{E}{|X|^2}$, the inequality following Eq. 36 holds. I want to thank Giuseppe Carleo for pointing out this subtlety.
}

We have that $\braket{\psi|O^2|\psi} \leq \sum_{\lambda} \lambda^2 |\braket{\lambda|\psi}|^2 \leq \lambda^2_{\textnormal{max}}$, where $\lambda_{\textnormal{max}}$ is the largest eigenvalue of $O$. For hermitian matrices, the largest eigenvalue coincides with the operator norm $\|O \| = \lambda_\textnormal{max}$, from which we have that $\braket{\psi|O^2|\psi} \leq \| O \|^2$. It follows that: $\sigma[X] \leq \|O\|^2$.

\begin{theorem}[Estimating Sparse Observables with AR]
There is an algorithm that estimates $\braket{\psi|O|\psi}$ for $\|O\| \leq 1$ to $\epsilon$-additive error with probability at least $1-e^{-n}$ using at most $32sn/\epsilon^2$, where $s$ is the row-sparsity of $O$.
\end{theorem}
\begin{proof}
Use the median of means estimator in Theorem~\ref{Thm:MoM}. One evaluation of the random variable $X$ costs at most $s$ queries, where $s$ is the row sparsity of $O$. The rest follows from Corollary~\ref{Cor:robust_fidelity}.
\end{proof}

\subsection{AR and dequantization?} SQ access was studied in dequantization and it is natural to ask if AR can lead to some improvements in that framework.
We outline mostly negative results for improvements of Ref.~\cite{Ewin18} using AR.

Ref.~\cite{Ewin18}, Proposition 4.2. gives an algorithm for estimating the inner product $\braket{x, y}$ of two real vectors $x$ and $y$, using the knowledge of the normalization constant of one of the vectors and error that depends on the normalization factor of both. Assuming $\braket{x,y} \geq 0$ (i.e. for non-negative vectors), the AR fidelity algorithm gives a good estimator of the inner product without knowledge of the normalizing factors. 

Proposition 4.3. of Ref.~\cite{Ewin18} crucially depends on computing the rejection sampling filter: 
\begin{align}
    r_i &= \frac{(Vw)^2(i)}{k \sum_j^k w_j^2 V_{ij}^2} = \left[k \sum_\ell^k \left( \sum_j^k \left[ \frac{V_{ij}}{V_{i \ell}} \frac{\omega_j}{\omega_\ell}\right] \right)^{-2} \right]^{-1}.
\end{align}
This can be computed with $O(k^2)$ AR ratio queries to $V$, {assuming that the entire matrix has been normalized with the same (possibly unknown) normalization factor}. This assumption is most likely too strong to be useful in the context of Tang's algorithm.

Lastly, the modified version of the FKV algorithm \cite{FKV04} (Algorithm~2 in \cite{Ewin18}) crucially relies on sampling from the distribution induced by norms of the matrix columns (normalization factors). There does not seem to be a simple way to circumvent this with the AR access. It would be really interesting to see if some of these limitations could be avoided in different dequantization algorithms.

\section{Postselection Gadgets}
\label{Sec:Gadgets}

Here we explore a different way of accessing the wavefunction that can be implemented with an NQS: postselection gadgets. Postselection gadgets are maps between NQSs that allow for a different set of conditional queries to the Born distribution encoded in the NQS called \emph{subcube conditional queries} in Ref.~\cite{Canonne19}. In contrast to PCOND, postselection gadgets cannot be instantiated efficiently for an arbitrary NQS. We use this property to show that there is an NQS with just three nodes that does not encode valid wavefunction and cannot be sampled from -- a counterpart to similar results for RBMs \cite{Long10, Martens13}. It is possible that the gadgets may have applications beyond this, but the analysis seems to be beyond reach of the techniques known to the author.

\subsection{Postselection Gadgets}

Let $|\psi_\theta(v | r)|^2 := p_\theta(v|r)$ where $v \in \lbrace -1, 1 \rbrace^n$ and $r \in \lbrace -1, 1, \star \rbrace^n$ be the distribution $|\psi_\theta (v)|^2 := p_{\theta}(v)$, conditioned on the event that for all non-$\star$ bits of $r$, the output bits of $v$ are fixed to bits of $r$. 
A \emph{postselection gadget} is a function that transforms a NQS $\theta$ to another NQS $\tilde{\theta}$ (possibly with additional nodes), such that $|\psi_\theta(v | r)|^2 = |\psi_{\tilde{\theta}}(v)|^2$.
\begin{theorem}
\label{Thm:conditionals}
There is a polynomial-time algorithm that, given $\theta = (a,b,W)$ and $r \in \lbrace 1, -1, \, \star \rbrace^n$, outputs $\tilde{\theta}$, such that $p_{\tilde{\theta}}(v) =  p_{\theta}(v|\, [r_i \neq \star] \implies [r_i = v_i]) = p_\theta(v|r)$.
\end{theorem}

\begin{proof} The following proof uses notation used in Sec.~\ref{Sec:NQS}.
For every non-$\star$ bit of $r$, introduce a hidden node $g$ and attach it to the corresponding visible node. 
Set the bias on this hidden node to $i(\pi/ 4)$ and couple it to the visible node with strength $-i(\pi/4) r_i$ (see Fig.~\ref{fig:proof-diagram}).  This gives:
\begin{equation}
\begin{aligned}
    f_{\tilde{\theta}}(v) 
    &= \left( \prod_{j: {r_j \neq \star}} 2 \cosh \left[i \frac{\pi}{4} (r_j-v_j) \right] \right) \times f_\theta(v).
\end{aligned}
\label{Eq:numerator_postselected}
\end{equation}
Let $F$ be the event\footnote{An event $F \subseteq \Omega$ is a subset of domain $\Omega$ of the probability distribution. } that $v_1 = r_1, v_2 = r_2 \ldots, v_k = r_k$.
Because $2\cosh \left(i \pi/2\right) =0$, we have that:
\begin{align}
    f_{\tilde{\theta}}(v \in F) &= 2^k f_\theta(v),  & f_{\tilde{\theta}}(v \not \in F)= 0. \label{Eq:conditions}
\end{align}
Let $E \subseteq \lbrace -1, +1 \rbrace^n$ be an arbitrary event. We have that:
\begin{equation}
\begin{aligned}
    p_{\tilde{\theta}}(E) &= p_{\tilde{\theta}}(E \cap F) + p_{\tilde{\theta}}(E \cap \neg F) \\ &= p_{\tilde{\theta}}(E | F) p_{\tilde{\theta}}(F) + p_{\tilde{\theta}}(E | \neg F)p_{\tilde{\theta}}(\neg F).
\end{aligned}
\label{Eq:tilde_conditional}
\end{equation}
From Eq.~\ref{Eq:conditions}, it follows that:
\begin{align}
    p_{\tilde{\theta}}(\neg F) &= \frac{\sum_{v\in \neg F} |f_{\tilde{\theta}}(v)|^2}{\sum_{v\in \lbrace -1, +1 \rbrace^N} |f_{\tilde{\theta}}(v)|^2} = 0.
\end{align}
and 
\begin{equation}
\begin{aligned}
    p_{\tilde{\theta}}(E|F) &= \frac{\sum_{v\in  E \cap F} |f_{\tilde{\theta}}(v)|^2}{\sum_{v\in F} |f_{\tilde{\theta}}(v)|^2}  = \frac{\sum_{v\in  E \cap F} 4^k |f_{\theta}(v)|^2}{\sum_{v\in F} 4^k |f_{\theta}(v)|^2} = p_\theta(E|F).
\end{aligned}
\end{equation}
Hence, from Eq.~\ref{Eq:tilde_conditional}: 
$p_{\tilde{\theta}}(E) = p_\theta(E|F)$. Thus, the output distribution of the augmented state is \textit{exactly} the conditional of the original distribution. See Fig.~\ref{fig:proof-diagram}.
\end{proof}

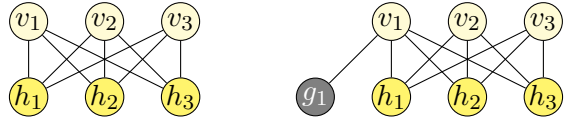
\begin{figure}[h]
\centering
\begin{tikzpicture}[baseline = (current bounding box.center)]
\node[diagram-node-light] at (0, 0)   (v1) {$v_1$};
\node[diagram-node-light] at (1, 0)   (v2) {$v_2$};
\node[diagram-node-light] at (2, 0)   (v3) {$v_3$};

\node[diagram-node-dark] at (0, -1)   (h1) {$h_1$};
\node[diagram-node-dark] at (1, -1)   (h2) {$h_2$};
\node[diagram-node-dark] at (2, -1)   (h3) {$h_3$};

\draw (h1) -- (v1);
\draw (h1) -- (v2);
\draw (h1) -- (v3);

\draw (h2) -- (v1);
\draw (h2) -- (v2);
\draw (h2) -- (v3);

\draw (h3) -- (v1);
\draw (h3) -- (v2);
\draw (h3) -- (v3);
\end{tikzpicture}~\hspace{1cm}~\begin{tikzpicture}[baseline = (current bounding box.center)]
\node[diagram-node-light] at (0, 0)   (v1) {$v_1$};
\node[diagram-node-light] at (1, 0)   (v2) {$v_2$};
\node[diagram-node-light] at (2, 0)   (v3) {$v_3$};

\node[diagram-node-dark] at (0, -1)   (h1) {$h_1$};
\node[diagram-node-dark] at (1, -1)   (h2) {$h_2$};
\node[diagram-node-dark] at (2, -1)   (h3) {$h_3$};
\node[diagram-node-inverted] at (-1, -1)   (g1) {$g_1$};

\draw (h1) -- (v1);
\draw (h1) -- (v2);
\draw (h1) -- (v3);

\draw (h2) -- (v1);
\draw (h2) -- (v2);
\draw (h2) -- (v3);

\draw (h3) -- (v1);
\draw (h3) -- (v2);
\draw (h3) -- (v3);

\draw (g1) -> (v1);
\end{tikzpicture}
\caption{To sample an RBM distribution over $\lbrace -1, 1\rbrace^3$ conditioned on $r = 1\star \star$, we can augment the RBM by an extra hidden node that only couples to the visible node $v_1$ with a bias $i \pi/4$ and coupling strength $-i \pi/4$.}
\label{fig:proof-diagram}
\end{figure}

Note that the function $\theta \mapsto \theta'$ is easy to compute. Postselection gadgets allow for sampling from a subset of conditional distributions of the distribution encoded into the NQS. 
Additional examples of postselection gadgets can be found in the appendix.  The construction almost trivially extends to Deep Boltzmann Machines \cite{Gao17, Carleo18}. While the analysis was inspired by its RBM counterpart by Long and Servedio in Ref.~\cite{Long10}, their gadget construction does not straightforwardly extend to RBMs.

\subsection{Not all NQS encode valid quantum states}
\label{Sec:Hardness}

 We show that many NQS do not encode valid quantum states. This follows from the fact that the postselection gadget allows postselection on probability zero events. Any NQS can be modified by adding two hidden nodes $g_1$ and $g_2$, as in Sec.~\ref{Sec:Gadgets}, that fix the value of the visible node $v_1$ to $1$ and $-1$ respectively. Eq.~\ref{Eq:numerator_postselected} then gives $f_{\tilde{\theta}}(v) = 0$, which also means that the encoded ``wavefunction'' is zero. The resulting NQS then does not encode any wavefunction because it cannot be normalized. The smallest NQS for which this works is one with two hidden nodes and one visible node. Such NQS naturally cannot be sampled from by any algorithm.
\begin{figure}[h]
\centering
\begin{tikzpicture}[baseline = (current bounding box.center)]
\node[diagram-node-light] at (0, 0)   (v1) {$v_1$};
\node[diagram-node-light] at (1, 0)   (v2) {$v_2$};
\node[diagram-node-light] at (2, 0)   (v3) {$v_3$};

\node[diagram-node-dark] at (0, -1)   (h1) {$h_1$};
\node[diagram-node-dark] at (1, -1)   (h2) {$h_2$};
\node[diagram-node-dark] at (2, -1)   (h3) {$h_3$};
\node[diagram-node-inverted] at (-1, -1)   (g1) {$g_1$};
\node[diagram-node-inverted] at (-2, -1)   (g2) {$g_2$};

\draw (h1) -- (v1);
\draw (h1) -- (v2);
\draw (h1) -- (v3);

\draw (h2) -- (v1);
\draw (h2) -- (v2);
\draw (h2) -- (v3);

\draw (h3) -- (v1);
\draw (h3) -- (v2);
\draw (h3) -- (v3);

\draw (g1) -> (v1);
\draw (g2) -> (v1);
\end{tikzpicture}
\caption{A NQS that does not encode a valid quantum state can be constructed by appending two auxiliary nodes, each of which fixes the value of $v_1$ to $+1$ and $-1$ respectively.}
\label{fig:zero_NQS}
\end{figure}
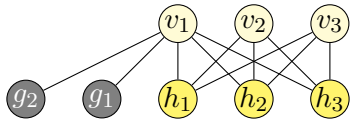

There does not seem to be any analogous \emph{simple} construction  for restricted Boltzmann machines with real valued coefficients. The reason is, as shown earlier, that RBMs cannot encode zeroes in the output probability because each outcome probability is lower-bounded by $2^m \exp(-\| a \|_\infty n)$ and the contribution of any ``gadget'' hidden node to the output RBM probability is a factor of $\cosh(x)$ for some real $x$, which is lower-bounded by $1$. 
We remark that existence of such ``zero-valued'' NQS is however not significant for applications in which the network is trained by sampling from output distributions of a sequence of NQSs. Some care must be however taken when encoding quantum states directly.

\subsection{Other applications?}
The postselection gadget can be in principle used for postselection without retries. However, even when the original NQS $\theta$ can be easily sampled from, it may be (and it is very likely in some cases) that the Gibbs sampling algorithm for the modified NQS $\tilde{\theta}$ won't be efficient anymore. Still, one toy example that may suggest that this could be more efficient than resampling (but that also seems useless) is the case in which the encoded distribution is a product distribution - fixing output bits of such distribution using gadgets will not lead to any slowdown.

In the case where the sampling algorithm remains efficient for the appropriate sequence of conditionals, one can also obtain a \emph{very crude} multiplicative approximation of the normalizing factor of the NQS wavefunction, essentially by retracing the RBM algorithm presented in Ref.~\cite{Long10}. The question of characterizing what conditional gadgets allow for efficient sampling however remains widely open.

Lastly, Ref.~\cite{Jonsson18} used a similar gadget construction to simulate universal random circuits with NQSs. The postselection gadgets from Thm.~\ref{Thm:conditionals} can be seen as extension of their result. See Appendix~\ref{Sec:Additional} for additional postselection gadgets and Appendix~\ref{Sec:PauliGadgets} for enconding of Pauli gates.

\section{Discussion}

We studied the access model offered by neural network quantum states (NQS), which, along with connection to previous results from conditional distribution testing, motivated the definition of amplitude ratio (AR) access model. We related AR to sample and query (SQ) access and showed that it retains some of the simulation capabilities of SQ. We gave some evidence that AR may be weaker than SQ and showed that existing results in distribution testing imply that AR is stronger than sample access. We then considered alternative access to the NQS wavefunctions by means of \emph{subcube conditional} queries and showed that even small NQS may not encode valid distributions. Our work left several questions open:

\begin{itemize} 
\item It would interesting to further explore the connections between AR and dequantization and understand if it is possible to meaningfully relax the SQ normalization requirement. 
\item Both definitions of CT and AR states assumed existence of a classical randomized sampler for the Born distribution. One may thus ask if there are any nontrivial states in the \emph{quantum-classical generalization} of CT and AR states, where we can efficiently sample from using quantum algorithm, but still classically estimate the ratios or amplitudes. Such states may be useful for construction of quantum algorithms based on conditional property testing results in Ref.~\cite{Canonne12, Chakraborty13}. It might be for example interesting to understand how this plays with some of the known supremacy results in which approximating a target amplitude is known to be $\# \P$-hard, yet there is an efficient quantum algorithm that samples the output \cite{Aaronson10}. It may be that, while the amplitudes are hard to compute, approximating their ratios remains tractable. 
\item It would be also interesting to understand, perhaps numerically, if and for what problems could the NQS postselection gadgets provide advantage when compared to postselection by resampling in some of the applications of NQS.
\end{itemize}

\section{Acknowledgements}
 I want to sincerely thank Ashley Montanaro and Noah Linden for their help. I also want to thank Giuseppe Carleo, Srini Arunachalam, Sergii Strelchuk, James Stokes and Juani Bermejo-Vega for their suggestions and discussion.

\printbibliography[heading=bibintoc]

\appendix

\section{Additional Postselection Gadgets}
\label{Sec:Additional}

\subsection{Hamming weight}

\begin{figure}[h]
\centering
\begin{tikzpicture}
\node[diagram-node-light] at (0, 0)   (v1) {$v_1$};
\node[diagram-node-light] at (1, 0)   (v2) {$v_2$};
\node[diagram-node-light] at (2, 0)   (v3) {$v_3$};

\node[diagram-node-dark] at (0, -1)   (h1) {$h_1$};
\node[diagram-node-dark] at (1, -1)   (h2) {$h_2$};
\node[diagram-node-dark] at (2, -1)   (h3) {$h_3$};
\node[diagram-node-inverted] at (-1, -1)   (g1) {$g_1$};

\draw (h1) -- (v1);
\draw (h1) -- (v2);
\draw (h1) -- (v3);

\draw (h2) -- (v1);
\draw (h2) -- (v2);
\draw (h2) -- (v3);

\draw (h3) -- (v1);
\draw (h3) -- (v2);
\draw (h3) -- (v3);

\draw (g1) -> (v1);
\draw (g1) -> (v2);
\draw (g1) -> (v3);
\end{tikzpicture}
\caption{To sample an RBM distribution conditioned on $\sum_i v_i = k$ for some $k 
\in [n]$, we use a two step process. We augment the RBM by an extra node coupled to all visible nodes with weight $W_k = \frac{i \pi}{k}$ and bias $b = -i \frac{\pi}{4}$. We do this all $n-1$ times, excluding only the sector we are interested in.}
\label{fig:proof-diagram2}
\end{figure}
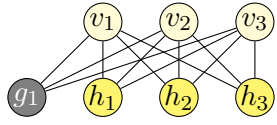 

Start with $n$ visible nodes. Add a single hidden node $g_k$ coupled to all of the visible nodes with weight $W = i\pi/n$ and $b = i \pi/2 - i k \pi/n$. This encodes: 
\begin{align}
    f_{k}(v) &= 2 \cosh \left[ \sum_i^n W v_i + b \right] = 2 \sin \left[ \frac{\pi}{n} \left( \sum_j^n v_j - k \right) \right].
\end{align}
Notice that $f_k(v) = 0$ if $\sum_j v_j = k$. By adding $(n-1)$ such nodes, we have that $F_k(v) :=  \prod_{j=1}^n f_j (v) / f_k(v)$. This function satisfies:
\begin{align}
    F_k(v) := \begin{cases} &0 \text{ if } \sum_i v_i \neq k, \\ 
    &2^{n-1} \sin \left[ \frac{\pi}{n} (|v|-1) \right]\sin \left[ \frac{\pi}{n} (|v|-2) \right] \ldots 
    \end{cases}
\end{align}
By modifying the visible biases, this can be normalized such that:
\begin{align}
F_k'(v) :=  \begin{cases} &0 \text{ if } \sum_i v_i \neq k. \\ 
    &c \text{ otherwise}.
    \end{cases}
\end{align}
for any constant $c$.

\subsection{Parity} Martens et. al showed in Ref.~\cite{Martens13} how to approximatelly  
represent parity (or any symmetric function) with an RBM of $m = n^2 +1$ hidden nodes. Here's an explicit construction that represents parity with only two hidden nodes of NQS: 

\begin{figure}[h]
\centering
\begin{tikzpicture}
\node[diagram-node-light] at (0, 0)   (v1) {$v_1$};
\node[diagram-node-light] at (1, 0)   (v2) {$v_2$};
\node[diagram-node-light] at (2, 0)   (v3) {$v_3$};

\node[diagram-node-dark] at (0, -1)   (h1) {$h_1$};
\node[diagram-node-dark] at (1, -1)   (h2) {$h_2$};
\node[diagram-node-dark] at (2, -1)   (h3) {$h_3$};

\node[diagram-node-inverted] at (-1, -1)   (g1) {$g_1$};
\node[diagram-node-inverted] at (-2, -1)   (g2) {$g_2$};

\draw (h1) -- (v1);
\draw (h1) -- (v2);
\draw (h1) -- (v3);

\draw (h2) -- (v1);
\draw (h2) -- (v2);
\draw (h2) -- (v3);

\draw (h3) -- (v1);
\draw (h3) -- (v2);
\draw (h3) -- (v3);

\draw (g1) -> (v1);
\draw (g1) -> (v2);
\draw (g1) -> (v3);

\draw (g2) -> (v1);
\draw (g2) -> (v2);
\draw (g2) -> (v3);
\end{tikzpicture}
\caption{Fixing parity with two hidden nodes.}
\label{fig:parity_with_two}
\end{figure}
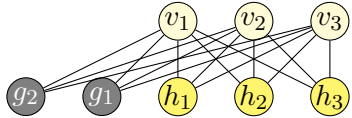 

Choose $W = i \pi/2$ for all edges from hidden to visible nodes and $\text{Im}(b) = 0$ for all hidden biases. This gives: 
\begin{align}
    f_\theta(v) := 4 \sin^2 \left( \frac{\pi}{2} \sum_i^n v_i \right) = \begin{cases}
    0 &\text{ if } \sum_i^n v_i \; \text{is even.} \\
    4 &\text{ if } \sum_i^n v_i \; \text{is odd}.
    \end{cases}
\end{align}
This represents a distribution that is uniform over all strings with odd parity and zero on the even parity ones and encodes parity distribution by only two hidden nodes. This gadget can be used to fix parity of NQS samples. Similar gadget has been used in Ref.~\cite{Carleo18}. 

\section{Pauli Gadgets}
\label{Sec:PauliGadgets}
We can use the same construction to postselection gadgets to simulate action of Pauli unitaries on the NQS wavefunction. This has application to simulation of quantum circuits in Ref.~\cite{Jonsson18}. The Pauli Z gadget appeared in \cite{Jonsson18}, but the others seem useful as a reference, so I will leave it here. Along with the $CZ$ gate and Hadamard approximation of \cite{Jonsson18}, the gadgets below give enconding of $X, Y, Z, H, CZ$ universal gateset, useful for simulating quantum circuits with NQS.

\paragraph{Pauli Z:} The action of a Pauli $Z_1$ operator applied to the first qubit on the NQS wavefunction can be understood as the following function on the amplitudes (recall that we chose $v \in \lbrace -1, +1\rbrace^n$):
\begin{align}
    f_\theta(v) \mapsto i^{(1+v_1)} f_\theta(v).
\end{align}
As an example, consider just two qubits, where we expect $Z_1$ to act as follows:
\begin{equation}
\begin{aligned}
    Z_1[f_\theta(++)] &= f_\theta(++), & Z_1[f_\theta(-+)] &= -f_\theta(-+), \\
    Z_1[f_\theta(-+)] &= f_\theta(-+), & Z_1[f_\theta(++)] &= -f_\theta(++).
\end{aligned}
\end{equation}
Let $Z_s = \bigotimes_{k = 1}^n Z^{s_k}$ for $s \in \lbrace 0,1 \rbrace^n$.
A gadget that applies this operator is defined as follows:
\begin{align}
    f_{Z_s}(v) &= \prod_{g \in [n]} \left[2\cosh\left(i \frac{\pi}{2}(1-v_g)\right)\right]^{s_k} \times f(v).
\end{align}
It can be constructed similarly to the postselection gadget by adding a hidden node to all nodes that should have Pauli gate applied to them, fixing their bias to $b = \left(i\frac{\pi}{2}\right)$ and interaction strength to $-i \frac{\pi}{2}$.

\paragraph{Pauli X:} The action of a Pauli $X_1$ operator applied to the first qubit on the NQS wavefunction can be understood as a permutation on the amplitudes. As an example, consider the action of $X_1$:

\begin{equation}
\begin{aligned}
    X_1[f_\theta(++)] &= f_\theta(-+), & X_1[f_\theta(-+)] &= f_\theta(++), \\
    X_1[f_\theta(-+)] &= f_\theta(++), & X_1[f_\theta(++)] &= f_\theta(-+).
\end{aligned}
\end{equation}

This transformation can be implemented by a map on the parameters, $\theta \mapsto \theta$, such that $a_1 \mapsto -a_1$ and $W_{1i} \mapsto -W_{1i}$ for all $i$.

\paragraph{Pauli Y:}
The gadget that applies this operator to the first qubit is defined as the action on the parameters as $a_1 \mapsto -a_1$ and $W_{1i} \mapsto -W_{1i}\; \forall i$ and an additional phase factor, similarly to \textbf{Pauli Z} gadget.

\paragraph{Pauli tensors:}
Combination of the above Pauli gadgets can be used to implement arbitrary Pauli tensor strings by adding up to $n$ hidden nodes.

\end{document}